\documentclass[a4paper,final]{scrartcl}

\usepackage[utf8]{inputenc}
\usepackage{microtype}
\usepackage[table]{xcolor}
\usepackage{wrapfig}
\usepackage{tikz}
\usepackage{amssymb}
\usepackage{tabularx}
\usepackage{xspace}

\usepackage{lmodern}
\usepackage{amsmath}
\usepackage{amssymb}
\usepackage{amsthm}
\usepackage{hyperref}

\renewcommand{\geq}{\geqslant}
\renewcommand{\leq}{\leqslant}
\renewcommand{\ge}{\geq}
\renewcommand{\le}{\leq}

\newcommand{\bigO}{\mathcal{O}}

\newcommand{\union}{\mathbin{\cup}}
\newcommand{\bigunion}{\mathop{\bigcup}}
\newcommand{\intersect}{\mathbin{\cap}}
\newcommand{\MonInt}{\ensuremath{\textsc{MonIsect}}}
\newcommand{\SgpInt}{\ensuremath{\textsc{SgpIsect}}}
\newcommand{\SAT}{\ensuremath{\textsc{3-Sat}}}

\newcommand{\abs}[1] {\ensuremath\left|#1\right|}
\newcommand{\ceil}[1] {\ensuremath\left\lceil#1\right\rceil}
\newcommand{\set}[2] {\ensuremath{\left\{#1 \mid #2\right\}}}
\newcommand{\os}[1] {\ensuremath{\left\{#1\right\}}}

\newcommand{\oL}{\mathbf{L}}

\newcommand{\AC}{\ensuremath{\mathsf{AC}}}
\newcommand{\qAC}{\ensuremath{\mathsf{qAC}}}

\newcommand{\NP}{\ensuremath{\mathsf{NP}}}
\newcommand{\PSPACE}{\ensuremath{\mathsf{PSPACE}}}
\newcommand{\vC}{\mathbf{C}}
\newcommand{\vV}{\mathbf{V}}
\newcommand{\vLV}{\mathbf{LV}}
\newcommand{\vLW}{\mathbf{LW}}
\newcommand{\vW}{\mathbf{W}}
\newcommand{\vM}{\mathbf{M}}

\newcommand{\vAp}{\mathbf{A}}
\newcommand{\vL}{\mathbf{L}}
\newcommand{\vDA}{\mathbf{DA}}
\newcommand{\vDS}{\mathbf{DS}}
\newcommand{\vI}{\mathbf{I}}
\newcommand{\vLI}{\mathbf{LI}}
\newcommand{\vLG}{\mathbf{LG}}
\newcommand{\vLDO}{\mathbf{LDO}}
\newcommand{\vLDS}{\mathbf{LDS}}
\newcommand{\vNil}{\mathbf{N}}
\newcommand{\vCom}{\mathbf{Com}}
\newcommand{\vDO}{\mathbf{DO}}
\newcommand{\vG}{\mathbf{G}}
\newcommand{\vD}{\mathbf{D}}
\newcommand{\vR}{\mathbf{R}}
\newcommand{\val}{\mathsf{val}}

\newcommand{\gR}{\mathcal{R}}
\newcommand{\gD}{\mathcal{D}}

\newcommand{\Z}{\mathbb{Z}}
\newcommand{\N}{\mathbb{N}}

\newcommand{\ie}{i.e.,~}
\newcommand{\eg}{e.g.~}
\newcommand{\ms}{\hspace*{0.5pt}}

\newcommand{\psp}{\hspace*{2pt}}
\newcommand{\cCP}{$\mathrm{const}$\psp{CP}\xspace}
\newcommand{\cPCP}{$\mathrm{const}$\psp{PCP}\xspace}
\newcommand{\plCP}{$\mathrm{polylog}$\psp{CP}\xspace}
\newcommand{\plPCP}{$\mathrm{polylog}$\psp{PCP}\xspace}
\newcommand{\pCP}{$\mathrm{poly}$\psp{CP}\xspace}
\newcommand{\pPCP}{$\mathrm{poly}$\psp{PCP}\xspace}

\makeatletter
\newcommand{\thickhline}{%
    \noalign {\ifnum 0=`}\fi \hrule height 1pt
    \futurelet \reserved@a \@xhline
}
\newcolumntype{"}{@{\hskip\tabcolsep\vrule width 1pt\hskip\tabcolsep}}
\newcolumntype{Y}{>{\centering\arraybackslash}X}

\makeatother

\theoremstyle{plain}
\newtheorem{theorem}{Theorem}
\newtheorem{lemma}[theorem]{Lemma}
\newtheorem{corollary}[theorem]{Corollary}
\newtheorem{proposition}[theorem]{Proposition}
\newtheorem{conjecture}[theorem]{Conjecture}

\title{The Intersection Problem for Finite Semigroups}

\author{Lukas Fleischer}
\date{FMI, University of Stuttgart\thanks{Supported by the German Research Foundation (DFG) under grant DI 435/5--2.} \\
  Universitätsstraße 38, 70569 Stuttgart, Germany\\
\texttt{fleischer@fmi.uni-stuttgart.de}}

\begin{document}

\maketitle

\begin{abstract}
  \noindent
  {\sffamily\normalsize\bfseries{Abstract.}} \
  We investigate the intersection problem for finite semigroups, which asks for
  a given set of regular languages, represented by recognizing morphisms to
  finite semigroups, whether there exists a word contained in their
  intersection.
  We introduce compressibility measures as a useful tool to classify the
  intersection problem for certain classes of finite semigroups into circuit
  complexity classes and Turing machine complexity classes. Using this
  framework, we obtain a new and simple proof that for groups and commutative
  semigroups, the problem is contained in $\NP$.
  We uncover certain structural and non-structural properties determining the
  complexity of the intersection problem for varieties of semigroups containing
  only trivial submonoids.
  More specifically, we prove $\NP$-hardness for classes of semigroups having a
  property called \emph{unbounded order} and for the class of all nilpotent
  semigroups of \emph{bounded order}. On the contrary, we show that bounded
  order and commutativity imply containment in the circuit complexity class
  $\qAC^k$ (for some $k \in \N$) and decidability in quasi-polynomial time.
  We also establish connections to the monoid variant of the problem.
\end{abstract}

\section{Introduction}

A careful analysis of the complexity of decision problems for regular languages
has triggered renewed interest in the classical intersection non-emptiness
problem (called \emph{intersection problem} in the following), as first
described by Kozen in 1977~\cite{koz77:short}, and in the closely related
\emph{membership problem for transformation
monoids}~\cite{BabaiLS87:short,Beaudry88thesis,BeaudryMT92,FurstHopcroftLuks80:short,Sims1967}.
The connection between these two problems stems from the observation that a set
of deterministic finite automata over a common alphabet can be considered as
transformations on the (disjoint) union of their states.
Both problems are well-known to be $\PSPACE$-complete in the general case but
become easier when the inputs are restricted to have certain structural
properties. These properties are often expressed in terms of membership to a
certain \emph{variety of finite monoids}; in the automaton setting, one
considers the transition monoids of the automata.
For example, for the variety of $\gR$-trivial monoids, usually denoted by
$\vR$, both problems were shown to be decidable in non-deterministic polynomial
time~\cite{BeaudryMT92}. On the other hand, it is known that
$\PSPACE$-completeness already holds for any variety not contained within
$\vDS$, the variety of all finite monoids whose regular $\gD$-classes form
subsemigroups.
However, for many subvarieties of $\vDS$, such as $\vL$ (the left-right dual of
$\vR$) or $\vDA$ (all aperiodic monoids from $\vDS$), the problems are only
known to be $\NP$-hard and to be contained within $\PSPACE$. The problem of
determining the exact complexity for varieties in this interval has been open
for more than 25~years~\cite{BeaudryMT92,tt02:short}.

Recently, Kuf\-leitner and the author suggested considering the algebraic
variant of the problem, where the languages in the input are represented by
finite monoids instead of automata~\cite{FleischerK18:short}. Formally, it is
defined as follows:
\vspace{1em}

\noindent
\begin{tabularx}{\textwidth}{p{1.5cm}X}
  \thickhline
  $\MonInt(\vC)$ \\
  \hline
  \textsf{Input}: & Morphisms $h_i \colon A^* \to M_i \in \vC$ and sets $P_i \subseteq M_i$ with $1 \le i \le k$ \\
  \textsf{Question}: & Is $h_1^{-1}(P_1) \intersect \cdots \intersect h_k^{-1}(P_k) \ne \emptyset$? \\
  \hline
\end{tabularx}
\vspace{1ex}

Here, $\vC$ is some fixed class of finite monoids and the monoids themselves
are given as multiplication tables.
Transitioning to the algebraic setting allowed for making some substantial
progress in understanding the complexity of the problem:
Kuf\-leitner and the author proved $\NP$-completeness of $\MonInt(\vDO)$ where
$\vDO$ is a quite large subvariety of $\vDS$ including both $\vL$ and $\vDA$.
Still, even for the monoid variant, $\PSPACE$-completeness is only known to
hold for varieties not contained within $\vDS$, a proper superset of $\vDO$.

Attempts to progress further in understanding the complexity of $\MonInt$ led
to the investigation of classes of semigroups $\vC$ instead of monoids:
\vspace{1em}

\noindent
\begin{tabularx}{\textwidth}{p{1.5cm}X}
  \thickhline
  $\SgpInt(\vC)$ \\
  \hline
  \textsf{Input}: & Morphisms $h_i \colon A^+ \to S_i \in \vC$ and sets $P_i \subseteq S_i$ with $1 \le i \le k$ \\
  \textsf{Question}: & Is $h_1^{-1}(P_1) \intersect \cdots \intersect h_k^{-1}(P_k) \ne \emptyset$? \\
  \hline
\end{tabularx}
\vspace{1ex}

As in the monoid variant, the semigroups are assumed to be given as
multiplication tables.
While making this distinction between monoids and semigroups may sound subtle
at first sight, it has a significant impact on complexity questions and is
expected to yield new insights.
For example, all known $\PSPACE$-hardness results rely heavily on the existence
of neutral letters.

We mainly investigate the intersection problem for \emph{varieties of finite
semigroups}.
In~\cite{FleischerK18:short}, $\MonInt(\vV)$ was shown to be $\NP$-hard for
every non-trivial variety of finite monoids $\vV$. Thus, in this work, we focus
on the intersection problem for varieties of finite semigroups containing only
trivial submonoids.
We describe an infinite sequence of varieties $\vV_1 \subseteq \vV_2 \subseteq
\cdots$ such that $\SgpInt(\vV_i) \in \AC^0$ for each $i \ge 1$ but the
intersection problem for its limit $\SgpInt(\vV_\infty)$ (where $\vV_\infty =
\bigcup_{i \in \N} \vV_i$) is $\NP$-complete.
This is surprising for the following reason: for the automaton and monoid
variants, all known hardness results are tied to purely structural properties.
$\NP$-hardness of $\MonInt$ comes from the fact that the problem is $\NP$-hard
even for the monoid $U_1$ and for the cyclic group $\Z / 2\Z$, and
$\PSPACE$-hardness comes from the fact that even $\MonInt(B_2^1)$ is
$\PSPACE$-hard~\cite{FleischerK18:short}.
Since every semigroup from $\vV_\infty$ is contained in infinitely many
varieties $\vV_k$ in the sequence above, the existence of such a pattern cannot
be the sole reason for $\NP$-hardness in the semigroup setting.
It is open whether a similar situation occurs below $\vDS$ in the monoid or
in the automaton setting.

To investigate other parameters with an impact on the complexity of the
problem, we introduce a versatile framework based on the notion of
\emph{product circuits properties}. These properties are a measure of
compressibility of witnesses for intersection non-emptiness.
Using this framework, we obtain a new and easy proof that both $\SgpInt(\vG)$
and $\SgpInt(\vCom)$ are contained in $\NP$.
We prove $\NP$-completeness of $\SgpInt$ for classes having a property we call
\emph{unbounded order} (this includes the class of all nilpotent and
commutative semigroups) and for the class of all nilpotent semigroups of
\emph{bounded order}.
On the contrary, we show that for every commutative variety with bounded order,
the intersection problem is contained in some uniform version of a circuit
complexity class $\qAC^k$ and thus decidable in quasi-polynomial time.  As
problems decidable in quasi-polynomial time cannot be $\NP$-hard unless the
exponential time hypothesis fails, this suggests that an interplay of
structural properties and non-structural properties determines the complexity
of the problem.
We also suggest a way to transfer complexity results from the monoid setting to
the semigroup setting.

\section{Preliminaries}
\label{sec:prelim}

\paragraph{Algebra.}

A \emph{semigroup} is a non-empty set equipped with an associative binary
operation, often also referred to as \emph{multiplication}.
A semigroup $M$ with a \emph{neutral element}, \ie{}an element $e \in M$ such
that $ex = x = xe$ for all $x \in M$, is called \emph{monoid}.
The neutral element is unique and usually denoted by $1$.
An element $x \in S$ is \emph{idempotent} if $x^2 = x$ and the set of all
idempotent elements of a semigroup $S$ is denoted by $E(S)$.
A \emph{zero} element $z$ of a finite semigroup $S$ satisfies $zx = z = xz$ for
all $x \in S$.
Each semigroup contains at most one zero element and a semigroup is
\emph{nilpotent} if its only idempotent element is a zero element.
The set of all finite words $A^*$ (resp.~all non-empty finite words $A^+$)
forms a monoid (resp.~semigroup) with concatenation as multiplication.

A \emph{subsemigroup} (resp.~\emph{submonoid}) of a semigroup (resp.~monoid) is
a subset closed under multiplication.
Let $S$ and $T$ be semigroups and let $M$ and $N$ be monoids.
The \emph{direct product} of $S$ and $T$ is the Cartesian product $S \times T$
with componentwise multiplication.
A \emph{semigroup morphism} from $S$ to $T$ is a mapping $h \colon S \to T$
such that $h(s)h(t) = h(st)$ for all $s, t \in S$.
A \emph{monoid morphism} from $M$ to $N$ is a semigroup morphism $h \colon M
\to N$ which additionally satisfies $h(1) = 1$.
The semigroup $T$ is a \emph{divisor} of $S$ if there exists a surjective
semigroup morphism from a subsemigroup of $S$ onto $T$.
The monoid $N$ is a \emph{divisor} of $M$ if there exists a surjective monoid
morphism from a submonoid of $M$ onto $N$.
We often use the term \emph{morphism} to refer to both semigroup and monoid
morphisms if the reference is clear from the context.
A morphism $h \colon A^+ \to S$ to a finite semigroup $S$ \emph{recognizes} a
language $L \subseteq A^+$ if $h^{-1}(P) = L$ for some set $P \subseteq S$. The
set $P$ is often called the \emph{accepting set} for $L$.

\paragraph{Varieties.}

A \emph{variety of finite semigroups} is a class of finite semigroups which is
closed under taking (semigroup) divisors and direct products.
A \emph{variety of finite monoids} is a class of finite monoids closed under
taking (monoid) divisors and direct products.
The class $\vG$ of all finite groups and the class $\vI$ containing only the
trivial semigroup $\os{1}$ are both varieties of finite semigroups and
varieties of finite monoids.
We also consider the following varieties of finite semigroups:

\begin{itemize}
  \item $\vCom$, the variety of all finite commutative semigroups,
  \item $\vNil$, the variety of all finite nilpotent semigroups,
  \item $\vAp_2 \cap \vNil$, the variety of all finite semigroups where $x^2 y
    = x^2 = y x^2$ for all $x, y \in S$,
  \item $\vLI_k$ (for $k \in \N$), the variety of all finite semigroups $S$
    which satisfy the equation $x_1 \cdots x_k z y_k \cdots y_1 = x_1 \cdots
    x_k y_k \cdots y_1$ for all $x_1, \dots, x_k, y_1, \dots, y_k, z \in S$.
\end{itemize}

Note that each of the varieties in this list contains semigroups which are not
monoids. Hence, they do not form varieties of finite monoids.
We will also briefly refer to the varieties $\vDS$ and $\vDO$ but their formal
definitions are not needed.

For a variety of finite semigroups $\vV$, we denote by $\vV_\vM$ the class of
all finite monoids which, when viewed as semigroups, belong to $\vV$. It is
easy to check that $\vV_\vM$ forms a variety of finite monoids.
For each semigroup $S$ and each idempotent element $e \in E(S)$, the set $eSe$
forms a monoid with the multiplication induced by $S$ and with neutral element
$e$, called the \emph{local monoid} at $e$.
For a variety of finite monoids $\vV$, we denote by $\vLV$ the variety of
finite semigroups whose local monoids belong to $\vV$.
The operators ${(\cdot)}_\vM$ and $\oL(\cdot)$ are closely related.

\begin{proposition}[Folklore]
  Let $\vV$ be a variety of finite monoids and let $\vW$ be a variety of finite
  semigroups.
  Then $\vW_\vM \subseteq \vV$ if and only if $\vW \subseteq \vLV$.
  In particular, $\vW \subseteq \vLW_\vM$.\label{prop:vm-lm}
\end{proposition}
\begin{proof}
  Suppose that $\vW_\vM \subseteq \vV$ and let $S$ be an arbitrary semigroup
  from $\vW$. For every idempotent element $e \in E(S)$, the monoid $eSe$ is a
  subsemigroup of $S$. By closure of $\vW$ under taking subsemigroups, we
  obtain $eSe \in \vW$. Since $eSe$ is a monoid, we obtain $eSe \in \vW_\vM$
  and by assumption, we have $eSe \in \vV$, as desired.

  Conversely, suppose that $\vW \subseteq \vLV$ and let $M$ be a monoid from
  $\vW$. Let $e$ be the identity element of $M$. Since $M \in \vLV$, we obtain
  $M = eMe \in \vV$.
\end{proof}

As a direct consequence, $\vLI$ is the largest variety of finite semigroups not
containing any non-trivial monoids. The following proposition connects $\vLI$
with the hierarchy of varieties $(\vLI_k)_{k \in \N}$ defined above.

\begin{proposition}[Folklore]
  Let $S$ be a finite semigroup of cardinality $n$. Then $S \in \vLI$ if and
  only if $S \in \vLI_{n + 1}$.
  \label{prop:li-nerb}
\end{proposition}
\begin{proof}
  Suppose that $S \in \vLI$ and let $x_1, \dots, x_{n+1}, y_1, \dots, y_{n+1},
  z \in S$. By the pigeon hole principle, there exist indices $i, i' \in \os{1,
  \dots, n+1}$ such that $i < i'$ and $x_1 \cdots x_i = x_1 \cdots x_{i'}$.
  Thus, $x_1 \cdots x_i e = x_1 \cdots x_i$ for $e = (x_{i+1}\cdots
  x_{i'})^\omega$ and for every $\omega \in \N$. In particular, we may choose
  $\omega$ such that $e$ is idempotent. Analogously, there exists some index $j
  \in \os{1, \dots, n+1}$ and some idempotent element $f$ such that $f y_j
  \cdots y_1 = y_j \cdots y_1$. Since $S \in \vLI$, we have $exf = ex(fef) =
  (exfe)f = ef = (eyfe)f = ey(fef) = eyf$ for all $x, y \in S$. This yields
  \begin{align*}
    x_1 \cdots x_{n+1} z y_{n+1} \cdots y_1 & = x_1 \cdots x_i e x_{i+1} \cdots x_{n+1} z y_{n+1} \cdots y_{j+1} f y_j \cdots y_1 \\
    & = x_1 \cdots x_i e x_{i+1} \cdots x_{n+1} y_{n+1} \cdots y_{j+1} f y_j \cdots y_1 \\
    & = x_1 \cdots x_{n+1} y_{n+1} \cdots y_1,
  \end{align*}
  which shows that $S \in \vLI_{n+1}$.

  Conversely, let $S$ be contained in $\vLI_{n+1}$. For all $e \in E(S)$ and
  for all $x \in S$, we have $exe = e^{n+1} x e^{n+1} = e^{n+1} e^{n+1} = e$
  where only the second equality uses $S \in \vLI_{n+1}$. Thus, every local
  monoid $eSe$ is trivial, and $S \in \vLI$.
\end{proof}

\paragraph{Complexity.}

We assume familiarity with standard definitions from circuit complexity.
A function has \emph{quasi-polynomial} growth if it is contained in
$2^{\bigO(\log^c n)} = n^{\bigO(\log^{c-1} n)}$ for some fixed $c \in \N$.
Throughout the paper, we denote by $\AC^k$ (resp.~$\qAC^k$) the class of
languages decidable by circuit families of depth $\bigO(\log^k n)$ and
polynomial size (resp.~quasi-polynomial size);
see~\cite{Barrington92,straubing94,Vollmer99} for details.
We allow NOT gates but do not count them when measuring the depth or the size
of a circuit.
We will also refer to the standard complexity classes $\NP$ and $\PSPACE$.
The \emph{exponential time hypothesis} states that a deterministic Turing
machine cannot decide $\SAT$ in subexponential time. If the hypothesis is true,
$\NP$-complete problems cannot be decided in quasi-polynomial time; see
\eg\cite{ImpagliazzoP99:short}.

\paragraph{Straight-Line Programs.}

A \emph{straight-line program} (\emph{SLP}) is a tuple $G = (V, A, P, X_s)$
where $V$ is a finite set of \emph{variables}, $A$ is a finite set of
\emph{letters}, $P \colon V \to (V \union A)^*$ is a mapping and $X_s \in V$ is
the so-called \emph{start variable} such that the relation
\begin{equation*}
  \set{(X, Y)}{P(X) \in (V \union A)^* Y (V \union A)^*}
\end{equation*}
is acyclic.
For a variable $X \in V$, the word $P(X)$ is the \emph{right-hand side} of $X$.
Starting with some word $\alpha \in (V \union A)^*$ and repeatedly replacing
variables $X \in V$ by $P(X)$ yields a word from $A^*$, the so called
\emph{evaluation of $\alpha$}, denoted by $\val(\alpha)$.
The word \emph{produced by $G$} is $\val(G) = \val(X_s)$.
If the reference to $A$ and $V$ is clear, we will often use the notation
$h(\alpha)$ instead of $h(\val(\alpha))$ for the image of the evaluation of a
word $\alpha \in (A \union V)^*$ under a morphism $h \colon A^+ \to S$.
Analogously, we write $h(G)$ instead of $h(\val(G))$.
The \emph{size} of $G$ is $\abs{G} = \sum_{X \in V} \abs{P(X)}$.
Each variable $X$ of an SLP~$G$ can be viewed as an SLP itself by making $X$
the start variable of $G$.

The \emph{canonical SLP of a word $w \in A^+$} is $G = (V, A, P, X_s)$ with $V
= \os{X_s}$ and $P(X_s) = w$.
The following simple lemma illustrates how SLPs can be used for compression;
see \eg\cite{CharikarLLPPSS05,FleischerK18:short} for a proof.

\begin{lemma}
  Let $G = (V, A, P, X_s)$ be an SLP and let $e \in \N$.
  Then there exists an SLP $H$ of size $\abs{H} \le \abs{G} + 4 \log(e)$ such
  that $\val(H) = (\val(G))^e$.
  \label{lem:slp-intro}
\end{lemma}

\section{Product Circuits Properties}
\label{sec:pcp}

Let $\vC$ be a class of finite semigroups and let $f \colon \N \to \N$ be a
monotonically increasing function.
We say that $\vC$ has the \emph{$f(n)$ circuits property} if for each
morphism $h_i \colon A^+ \to S$ to a finite semigroup $S \in \vC$
and for each $w \in A^+$, there exists an SLP $G$ of size at most $f(\abs{S})$
such that $h(G) = h(w)$.
We say that $\vC$ has the \emph{$f(n)$ product circuits property} if for each
set of morphisms $h_i \colon A^+ \to S_i$ to finite semigroups $S_1, \dots, S_k
\in \vC$ and for each $w \in A^+$, there exists an SLP $G$ of size at most
$f(\abs{S_1} + \dots + \abs{S_k})$ such that $h_i(G) = h_i(w)$ for all $i \in
\os{1, \dots, k}$.
For a class of functions $\mathcal{C}$, we say that $\vC$ has the
\emph{$\mathcal{C}$ circuits property} (resp. \emph{$\mathcal{C}$~product
circuits property}) if $\vC$ has the $f(n)$ circuits property (resp.~$f(n)$
product circuits property) for some $f \in \mathcal{C}$.

Let us introduce some abbreviations for commonly used classes of functions. We
will use the terms
\begin{itemize}
  \item \emph{constant circuits property} and \emph{constant product circuits
    property} (\emph{\cCP} and \emph{\cPCP}, in short) for the class of constant
    functions, \ie{}the class of all functions of the form $f(n) = c$ for some
    $c \in \N$,
  \item \emph{poly-logarithmic circuits property} and \emph{poly-logarithmic
    product circuits property} (\emph{\plCP} and \emph{\plPCP}, in short) for
    the class of poly-logarithmic functions, \ie{}the class of all functions
    $f(n) = \log^c n$ for some $c \in \N$, and
  \item \emph{polynomial circuits property} and \emph{polynomial product
    circuits property} (\emph{\pCP} and \emph{\pPCP}, in short) for the class
    of polynomials, \ie{}the class of all functions of the form $f(n) = n^c$
    for some $c \in \N$.
\end{itemize}

The intuition behind these concepts is as follows. The $f(n)$ circuits property
is a compressibility measure for witnesses of non-emptiness of a language given
by a recognizing morphism.
The $f(n)$ product circuits property is a compressibility measure for witnesses
of non-emptiness of intersections of languages given by recognizing morphisms.
The terminology is inspired by the \emph{poly-logarithmic circuits property}
which was introduced in~\cite{FleischerCCC18:short}: having the $f(n)$ circuits
property is equivalent to requiring every element of a subsemigroup $S$ of a
semigroup from the class to be computable by an \emph{algebraic circuit} of
size $f(n)$ over any set of generators of $S$. Analogously, having the $f(n)$
product circuits property can be expressed in terms algebraic circuits with
multiplication gates for the direct product of semigroups.
It is clear that the $f(n)$ product circuits property implies the $f(n)$
circuits property. For the other direction, a weaker statement holds.

\begin{proposition}
  Let $\vC$ be a class of finite semigroups which is closed under taking direct
  products and has the $f(n)$ circuits property.
  Then $\vC$ has the ${f(n^n)}$ product circuits property.
  \label{prop:cp-pcp}
\end{proposition}
\begin{proof}
  Suppose we are given morphisms $h_i \colon A^+ \to S_i$ to finite
  semigroups $S_1, \dots, S_k \in \vC$ and a word $w \in A^+$. Let $N =
  \abs{S_1} + \cdots + \abs{S_k}$. Every semigroup contains at least one
  element, so $N^N \ge N^k$ is an upper bound for the product $\abs{S_1} \cdots
  \abs{S_k}$.

  Let $S$ be the direct product $S_1 \times \cdots \times S_k$ and let $h
  \colon A^+ \to S$ be the morphism defined by $h(a) = (h_1(a), \dots, h_k(a))$
  for all $a \in A$.
  By closure of $\vC$ under taking direct products, we have $S \in \vC$.
  Since $\vC$ has the $f(n)$ circuits property, there exists some SLP $G$ of
  size at most $f(\abs{S}) = f(\abs{S_1} \cdots \abs{S_k}) \le f(N^N)$ such
  that $h(G) = h(w)$. By construction, $h_i(G) = h_i(w)$ for all $i \in \os{1,
  \dots, k}$.
\end{proof}

An essential ingredient in the proof of $\MonInt(\vDO) \in \NP$ is that the
variety of finite groups $\vG$ has the \pPCP.
In~\cite{FleischerK18:short}, this was verified by analyzing a variant of the
Schreier-Sims algorithm.
Using the previous proposition, we obtain a much simpler proof:
it is well known\,---\,and easy to show\,---\,that $\vG$ has the {\plCP}, a
result often called the \emph{Babai-Szemer{\'e}di Reachability
Lemma}~\cite{BabaiS84}. The statement then follows from the following corollary
of Proposition~\ref{prop:cp-pcp}.

\begin{corollary}
  Let $\vC$ be a class of finite semigroups which is closed under taking direct
  products and has the {\plCP}. Then $\vC$ has the {\pPCP}.
  \label{crl:plcp-pcp}
\end{corollary}

The corollary also implies that the variety of all commutative semigroups,
which was shown to have the \plCP in~\cite{FleischerCCC18:short}, has the
\pPCP.

Circuits properties and product circuits properties have a big impact on the
complexity of the so-called \emph{Cayley semigroup membership problem} and the
intersection problem for a given class.
The remainder of this section is devoted to establishing this link for product
circuits properties.

\begin{lemma}
  Let $h \colon A^+ \to S$ be a morphism to a finite semigroup $S$ of size $N$
  and let $G$ be an SLP of size $m$ over $A$.
  Then there exists an unbounded fan-in Boolean circuit of size $m(N^2 +
  \abs{A} + 2) \ceil{\log N}$ and depth $2m + 2$ which computes $h(G)$.
  Given the SLP, this circuit can be computed by a deterministic Turing machine
  in time polynomial in the circuit size.
  \label{lem:slp}
\end{lemma}
\begin{proof}
  Single multiplications can be performed by circuits of size $(N^2 + 1)
  \ceil{\log N}$ with one layer of AND gates and one layer of OR gates:
  to perform a multiplication of two elements $x$ and $y$, we need to extract
  the $\ceil{\log N}$-bit entry of the multiplication table in row $x$ and
  column $y$.
  We create a layer of $N^2 \ceil{\log N}$ AND gates, followed by a layer of
  $\ceil{\log N}$ OR gates.
  Each AND gate is connected to one bit of the multiplication table in the input
  and to all bits of the values $x$ and $y$. Some of the incoming wires
  corresponding to the values $x$ and $y$ are negated such that the AND gate
  copies the bit of the multiplication table if it belongs to the corresponding
  entry $(x, y)$ and evaluates to $0$ otherwise.
  In the second layer, there are $\ceil{\log N}$ OR gates. The $k$-th of these OR
  gates is fed with the outputs of all AND gates corresponding to the $k$-th bit
  of some multiplication table entry. Thus, there are $N^2$ incoming wires to
  each OR gate.
  Since, for given input values $x$ and $y$, at most one of the incoming wires to
  each OR gate evaluates to $1$, the result of the product $x \cdot y$ then
  clearly appears as output value of the OR gates.

  A very similar layout is used to lookup the image of a letter $a \in A$ under
  the morphism $h \colon A^+ \to S$. First, $\abs{A} \ceil{\log N}$ AND-gates are
  used to zero out the images of all letters except for the image of the letter
  $a$. Then, we use $\ceil{\log N}$ OR gates to perform a bitwise OR of all these
  preprocessed images. Since all images except $h(a)$ are zeroed out, the result
  is $h(a)$, as desired.

  We evaluate the image of each of the variables bottom-up:
  for all letters $a \in A$ occurring in $G$ we first compute the image $h(a)$.
  Then, if $P(X) = \gamma_1 \cdots \gamma_\ell$ for some $\gamma_1, \dots,
  \gamma_\ell \in V \union A$ and the images $h(\gamma_1), \dots,
  h(\gamma_\ell)$ have already been computed, we compute $h(X) = h(\gamma_1)
  \cdots h(\gamma_\ell)$ by performing $\ell-1$ multiplications.

  Clearly, each ``lookup gadget'', each multiplication gadget and the wires
  connecting these components can be computed by a deterministic Turing machine
  in time polynomial in the size of the resulting circuit.
\end{proof}

We are now able to prove the main result of this section.

\begin{theorem}
  Let $\vC$ be a class of finite semigroups with the $f(n)$ product circuits
  property. Then $\SgpInt(\vC)$ is decidable by a family of unbounded fan-in
  Boolean circuits of size $\bigO((f(n) + n)^{(f(n))^2} \ms f(n) \ms n^3 \log
  n)$ and depth $\bigO(f(n))$. For each input size $n \in \N$, the
  corresponding circuit can be computed by a deterministic Turing machine in
  time polynomial in the size of the resulting circuit.
  \label{thm:circuits}
\end{theorem}
\begin{proof}
  Suppose we are given morphisms $h_i \colon A^+ \to S_i$ to finite semigroups
  $S_i \in \vC$ and sets $P_i \subseteq S_i$ where $1 \le i \le k$ for some $k
  \in \N$.
  We let $N = \abs{S_1} + \dots + \abs{S_k}$. Note that if $n$ denotes the
  input size of the $\SgpInt$ instance, we have $N \le n$ and $\abs{A} \le n$.
  Since $\vC$ has the $f(n)$ product circuits property, we know that if there
  exists a word $w \in A^+$ such that $h_i(w) \in P_i$ for all $i \in \os{1,
  \dots, k}$, then this word is generated by some SLP of size at most $f(N) \le
  f(n)$.

  First, note that for a given fixed SLP of size $f(n)$, we can compute the
  image of the word generated by the SLP under each of the morphisms by an
  unbounded fan-in Boolean circuit of size $\bigO(n \ms f(n) \ms n^2 \log n)$
  and depth $\bigO(f(n))$ by Lemma~\ref{lem:slp}.
  Since there are at most $((f(n)+n)^{f(n)})^{f(n)}$ different SLPs of size
  $f(n)$\,---\,at most $f(n)$ variables and at most $(f(n)+n)^{f(n)}$ possible
  right-hand sides per variable\,---\,we can do this evaluation for each of the
  SLPs in parallel, check whether any of them produces a witness for
  intersection non-emptiness and feed the outcomes of all the circuits into a
  single OR gate.
  It is clear that an enumeration of all SLPs of size at most $f(n)$ can be
  realized by a deterministic Turing machine in time polynomial in the output
  size.
\end{proof}

For classes with the \cPCP and classes with the \plPCP, efficient decidability
of $\SgpInt$ is an immediate consequence.

\begin{corollary}
  Let $\vC$ be a class of finite semigroups with the \cPCP.
  Then the decision problem $\SgpInt(\vC)$ is contained in $\AC^0$.
\end{corollary}

\begin{corollary}
  Let $\vC$ be a class of finite semigroups with the \plPCP.
  Then the decision problem $\SgpInt(\vC)$ is contained in $\qAC^k$ for some $k
  \in \N$.
  Moreover, it is decidable in quasi-polynomial time and thus not $\NP$-hard,
  unless the exponential time hypothesis fails.%
  \label{crl:plpcp}
\end{corollary}
\begin{proof}
  Containment in $\qAC^k$ is an immediate consequence of
  Theorem~\ref{thm:circuits}. For decidability in quasi-polynomial time, we can
  use a Turing machine that first computes and then evaluates the circuit. The
  circuit evaluation is done by computing the output value of a gate whenever
  all its inputs are available.
\end{proof}

For the \pPCP, the statement of Theorem~\ref{thm:circuits} only yields
exponential-size circuits. We restate a more useful complexity result on \pPCP
classes from~\cite{FleischerK18:short}.

\begin{theorem}
  Let $\vC$ be a class of finite semigroups with the \pPCP.
  Then $\SgpInt(\vC)$ is contained in $\NP$.
  \label{thm:ppcp}
\end{theorem}
\begin{proof}
  We proceed as in the proof of Theorem~\ref{thm:circuits} but instead of
  generating a circuit evaluating all SLPs of polynomial size in parallel, we
  non-deterministically guess only one such SLP. We then evaluate the
  corresponding circuit in polynomial time as described in Lemma~\ref{lem:slp}.
\end{proof}

Together with the observations above, we obtain an easy proof of containment of
both $\SgpInt(\vG)$ and $\SgpInt(\vCom)$ in $\NP$.

Even though product circuits properties are a powerful tool, in some cases, it
is sufficient to consider short witnesses without compression.
This is particularly true for varieties not containing any subgroups which we
shall mostly be concerned with in the following section.
Moreover, for the \cPCP, compressibility and the existence of short
(non-compressed) witnesses are actually equivalent.

\begin{proposition}
  A class of finite semigroups $\vC$ has the \cPCP if and only if there exists
  some constant $\ell \in \N$ such that every non-empty intersection of
  languages recognized by semigroups from $\vC$ contains a word of length at
  most $\ell$.
  \label{prop:cpcp}
\end{proposition}
\begin{proof}
  The direction from right to left is trivial: every word $w$ of length at most
  $\ell$ can be represented by its canonical SLP, which then has size at most
  $\ell$ as well.

  For the converse direction, suppose that there exists some $s \in \N$ such
  that every non-empty intersection contains a word generated by an SLP of size
  at most $s$. It is easy to see that the length of such a word is at most
  $s^s$: there are at most $s$ variables and the right-hand side of every
  variable has length at most $s$; the claim now follows by induction. Thus, we
  obtain the desired statement by setting $\ell = s^s$.
\end{proof}

\section{The Intersection Problem for Locally Finite Semigroups}
\label{sec:complexity}

Before presenting any algorithms and hardness results for $\SgpInt$, let us
first describe how to transfer existing results to the semigroup setting.

\begin{proposition}
  Let $\vV$ be a variety of finite semigroups. If $\vV \not\subseteq \vLI$,
  then $\SgpInt(\vV)$ is $\NP$-hard. If $\vV \not\subseteq \vLDS$, then
  $\SgpInt(\vV)$ is $\PSPACE$-hard.
  \label{prop:hardness}
\end{proposition}
\begin{proof}
  If $\vV \not\subseteq \vLI$, then $\vV_\vM \not\subseteq \vI$. Therefore,
  by~\cite[Theorem 8]{FleischerK18:short}, $\MonInt(\vV_\vM)$ is $\NP$-hard.
  The claim now follows from the fact that $\MonInt(\vV_\vM)$ is trivially
  $\AC^0$-reducible to $\SgpInt(\vV)$.
  The same technique allows lifting $\PSPACE$-hardness of $\MonInt(\vV_\vM)$
  in the case $\vV_\vM \not\subseteq \vDS$~\cite[Theorem
  11]{FleischerK18:short}.
\end{proof}

It seems plausible that the $\vL(\cdot)$ operator can be used to lift
complexity results from $\MonInt$ to $\SgpInt$ in a more general way. We thus
conjecture:

\begin{conjecture}
  If $\MonInt(\vV)$ is in $\NP$, then $\SgpInt(\vLV)$ is in $\NP$.
  \label{con:transfer}
\end{conjecture}

By~\cite{FleischerK18:short}, a proof of this conjecture would immediately
yield that $\SgpInt(\vLDO)$ is contained in $\NP$.
A possible approach is making use of the fact that for a local variety of
finite monoids $\vV$, we have $\vLV = \vV * \vD$; see \eg\cite{str85jpaa} for
details. However, one also needs to account for the size of semigroups from
$\vV * \vD$.
Surprisingly, even lifting the group case is much harder than one would expect.
Our attempts to adapt the group algorithm from~\cite{FleischerK18:short} failed
and it is known from~\cite{FleischerCCC18:short} that $\vLG$ does not have the
\plCP, so we cannot use Corollary~\ref{crl:plcp-pcp} as in the group case.

To summarize, up to this point, the complexity landscape of $\SgpInt$ looks as
follows. By Proposition~\ref{prop:hardness}, the problem is $\NP$-hard for
every variety $\vV \not\subseteq \vLI$.
Using the $\vDO$-algorithm from~\cite{FleischerK18:short}, we know that the
problem is $\NP$-complete for every variety $\vV \subseteq \vDO$ not contained
within $\vLI$.
For $\vV \not\subseteq \vLDS$, the problem is $\PSPACE$-complete.
This leaves two classes of varieties for further investigation:
\begin{enumerate}
  \item For $\vV \not\subseteq \vDO$ and $\vV \subseteq \vLDS$, we do not know
    whether the problem is always $\NP$-complete, whether it becomes
    $\PSPACE$-complete for varieties contained within $\vLDS$ already and
    whether any other classes inside $\vLDS$ are connected to natural
    complexity classes, such as the polynomial hierarchy.
  \item Thus far, we do not have any hardness results for $\vV \subseteq \vLI$.
\end{enumerate}
The remainder of this section is devoted to the second class of varieties. On
one hand, it is not difficult to see that $\SgpInt(\vLI)$ is contained in
$\NP$. On the other hand, $\NP$-hardness holds only for some subvarieties of
$\vLI$ but not for others.
Containment in $\NP$ actually already follows from $\vLI \subseteq \vDO$ but it
also is an immediate consequence of the following result.

\begin{theorem}
  For each $k \ge 1$, the variety $\vLI_k$ has the $2k$ product circuits
  property.
  In particular, $\SgpInt(\vLI_k)$ is contained in $\AC^0$.
  \label{thm:nerb}
\end{theorem}
\begin{proof}
  It suffices to show that for each $k \in \N$ and for each finite semigroup $S
  \in \vLI_k$, each morphism $h \colon A^+ \to S$ and each $u = a_1 \cdots
  a_\ell \in A^+$ with $\ell > 2k$, the word $v = a_1 \cdots a_k a_{\ell-k+1}
  \cdots a_\ell$ satisfies $h(v) = h(u)$. To see this, note that
  \begin{align*}
    h(v) & = h(a_1 \cdots a_k a_{\ell-k+1} \cdots a_\ell) = h(a_1) \cdots h(a_k) h(a_{\ell-k+1}) \cdots h(a_\ell) \\
         & = h(a_1) \cdots h(a_k) h(a_{k+1} \cdots a_{\ell-k}) h(a_{\ell-k+1}) \cdots h(a_\ell) = h(a_1 \cdots a_\ell) = h(u)
  \end{align*}
  where the third equality holds by the definition of $\vLI_k$.
  The length of $v$ is $\abs{v} = k + (\ell - (\ell - k)) = 2k$.
  Since the word $v$ does not depend on $h$ or on $S$, the canonical SLP of $v$
  yields the desired product circuits property.
\end{proof}

Combining Theorem~\ref{thm:nerb} with Proposition~\ref{prop:li-nerb}, we
immediately obtain that $\vLI$ has the $2n + 2$ product circuits property: each
of the semigroups $S_1, \dots, S_k$ in the input has cardinality at most $N =
\abs{S_1} + \cdots + \abs{S_k}$. Hence, all semigroups $S_i$ belong to the
variety $\vLI_{N+1}$ and there exists a witness of size at most $2N+2$.

\begin{corollary}
  The variety $\vLI$ has the {\pPCP}. In particular, $\SgpInt(\vLI)$ is
  contained in $\NP$.
  \label{crl:li}
\end{corollary}

Another consequence of Proposition~\ref{prop:li-nerb} is $\bigunion_{k \in \N}
\vLI_k = \vLI$. For each variety in the infinite sequence $\vLI_1 \subseteq
\vLI_2 \subseteq \cdots$, the intersection problem is in~$\AC^0$ but for its
limit $\vLI$, the problem is only contained in $\NP$\,---\,and it is actually
$\NP$-complete, as we shall see later.
Therefore, in contrast to previously obtained hardness results which relied on
purely structural properties, other parameters interfere with the complexity of
$\SgpInt$ below $\vLI$.
We will investigate this phenomenon more carefully.
A semigroup is \emph{monogenic} if it is generated by a single element.
The \emph{order} of a class $\vC$ of finite semigroups is the supremum of the
cardinalities of all monogenic subsemigroups contained in $\vC$. If the order
is~$\infty$, the class is said to have \emph{unbounded order}.
The following observation will be used implicitly several times later.

\begin{lemma}[Folklore]
  Let $S$ be a finite semigroup from $\vLI$ and let $s \in S$. Then there
  exists some integer $n \in \N$ such that for all $i \in \N$, we have $s^{n+i}
  = s^n$. This integer is the order of the monogenic subsemigroup generated by
  $s$. Moreover, if $S$ is nilpotent, then $s^n$ is the zero element.
  \label{lem:order-li}
\end{lemma}
\begin{proof}
  Since $S$ is finite, there exist $n \in \N$ and $p \ge 1$ with $s^n =
  s^{n+p}$. Let $n$ and $p$ be minimal with this property. If $p > 1$, then
  $s^{np+1}$ generates a non-trivial subgroup of $S$, a contradiction to the
  assumption that $S \in \vLI$. Thus $p = 1$, yielding the first part of the
  statement.

  It is clear that $s^{2n} = s^{n+n} = s^n$, thus $s^n$ is idempotent. Since in
  a nilpotent semigroup, the only idempotent element is a zero element, we
  obtain the desired statement.
\end{proof}

In follow-up results, we will use reductions from $\SAT$ to prove
$\NP$-hardness of $\SgpInt$ for varieties of semigroups with certain
properties.  To simplify notation, let us introduce some definitions.
For a set of \emph{variables} $X = \os{x_1, \dots, x_k}$, we let $\overline X =
\set{\overline x}{x \in X}$ where each $\overline x$ is a new symbol.
The set of \emph{literals} over $X$ is $X \union \overline X$ and a set of
literals is a \emph{clause}.
An \emph{assignment} $\mathcal{A} \colon X \to \os{0, 1}$ of truth values to
the variables $X$ can be extended to all literals over $X$ by letting
$\mathcal{A}(\overline x) = 1 - \mathcal{A}(x)$ and to clauses $C \subseteq X
\union \overline X$ by letting $\mathcal{A}(C) =
\max\set{\mathcal{A}(\ell)}{\ell \in C}$.
An assignment $\mathcal{A}$ \emph{satisfies} a set of clauses $\os{C_1, \dots,
C_n}$ if $\mathcal{A}(C_j) = 1$ for all $j \in \os{1, \dots, n}$.
For a word $w \in (X \union \overline X)^+$, the mapping $\mathcal{A}_w \colon
X \to \os{0, 1}$ defined by $\mathcal{A}_w(\ell) = 1$ if and only if $w \in (X
\union \overline X)^* \ell (X \union \overline X)^*$ for all $\ell \in X \union
\overline X$ is called the \emph{assignment induced by $w$}. Note that this
assignment is well-defined whenever $\os{w} \intersect (X \union \overline X)^*
x_i (X \union \overline X)^* \intersect (X \union \overline X)^* \overline{x_i}
(X \union \overline X)^* = \emptyset$ for all $i \in \os{1, \dots, k}$.
Conversely, for a given assignment $\mathcal{A} \colon X \to \os{0, 1}$, we
call $w_\mathcal{A} = \ell_1 \cdots \ell_k$, where $\ell_i = x_i$ if
$\mathcal{A}(x_i) = 1$ and $\ell_i = \overline{x_i}$ otherwise, the \emph{word
induced by $\mathcal{A}$}.

\begin{theorem}
  If $\vV$ is a variety of finite semigroups with unbounded order, then the
  decision problem $\SgpInt(\vV)$ is $\NP$-hard.
  \label{thm:unbounded}
\end{theorem}
\begin{proof}
  We may assume $\vV \subseteq \vLI$, otherwise $\SgpInt(\vV)$ is $\NP$-hard
  by~Proposition~\ref{prop:hardness}.
  For each $k \in \N$ the semigroup $S_k = \os{1, \dots, k}$ with the binary
  operation $i \circ j = \min\os{i+j,k}$ is contained in~$\vV$.
  To see this, take some arbitrary $k \in \N$. Since $\vV$ has unbounded order,
  some monogenic semigroup $T$ of cardinality $m \ge k$ appears as a
  subsemigroup in $\vV$. Let $s$ be a generator of $T$. By
  Lemma~\ref{lem:order-li} and since $m \ge k$, the mapping $h \colon T \to
  S_k$ defined by $h(s) = 1$ is a surjective morphism. By closure of $\vV$
  under divisors, the semigroup $S_k$ itself belongs to $\vV$.

  We now reduce $\SAT$ to $\SgpInt(\vV)$.
  Suppose we are given a set of variables $X = \os{x_1, \dots, x_k}$ and a set
  of clauses $\os{C_1, \dots, C_n}$ where $C_j = \os{\ell_{j1}, \ell_{j2},
  \ell_{j3}}$ for each $j \in \os{1, \dots, n}$ and for literals $\ell_{j1},
  \ell_{j2}, \ell_{j3}$ over $X$.

  We let $S = S_{k+2}$ be the monogenic semigroup of cardinality $k + 2$
  defined above.
  We introduce morphisms $g_0, \dots, g_k, h_1, \dots, h_n \colon (X \union
  \overline X)^+ \to S$ defined by
  \begin{equation*}
    g_i(\ell) =
      \begin{cases}
        2 & \text{if $i > 0$ and $\ell \in \os{x_i, \overline{x_i}}$}, \\
        1 & \text{otherwise},
      \end{cases}
    \qquad
    h_j(\ell) =
      \begin{cases}
        2 & \text{if $\ell \in C_j$}, \\
        1 & \text{otherwise}.
      \end{cases}
  \end{equation*}
  for $0 \le i \le k$ and $1 \le j \le n$.
  We let $P_0 = \os{k}$, $P_1 = \dots = P_k = \os{k+1}$ and $Q_1 = \dots =
  Q_n = \os{k+1, k+2}$.
  It is easy to check that the intersection
  \begin{equation*}
    L = \bigcap_{i = 0}^n g_i^{-1}(P_i) \intersect \bigcap_{j = 1}^k h_j^{-1}(Q_j)
  \end{equation*}
  is non-empty if and only if there exists a satisfying assignment.
  To see this, the following three observations are crucial:
  \begin{enumerate}
    \item $g_0^{-1}(P_0)$ contains all words over $(X \union \overline X)$ with
      exactly $k$ letters,
    \item $g_i^{-1}(P_i) \intersect g_0^{-1}(P_0)$ contains all words from the
      set $(X \union \overline X)^k$ with exactly one occurrence of $x_i$ or
      exactly one occurrence of $\overline{x_i}$ (but not both), and
    \item $h_j^{-1}(Q_j) \intersect g_0^{-1}(P_0)$ contains all words from the
      set $(X \union \overline X)^k$ with at least one occurrence of any of the
      literals $\ell_{j1}, \ell_{j2}, \ell_{j3}$.
  \end{enumerate}
  By the first two properties, all words from $L$ are of the form $\ell_1
  \cdots \ell_k \in (X \union \overline X)^k$ with $\abs{\os{\ell_1, \dots,
  \ell_k} \intersect \os{x_i, \overline{x_i}}} = 1$ for all $i \in \os{1,
  \dots, k}$.
  Thus, for each $w \in L$, the assignment $\mathcal{A}_w$ induced by $w$ is
  well-defined.

  Now, if $w \in L$, by the third property, we have $\mathcal{A}_w(\ell_{j1}) =
  1$ or $\mathcal{A}_w(\ell_{j2}) = 1$ or $\mathcal{A}_w(\ell_{j3}) = 1$ for
  each $j \in \os{1, \dots, n}$. Thus, $\mathcal{A}_w$ is satisfying.
  Conversely, if there exists a satisfying assignment $\mathcal{A} \colon X \to
  \os{0, 1}$, the word induced by $\mathcal{A}$ is contained in $L$.

  It is obvious that the reduction can be performed in polynomial time. A more
  careful analysis shows that the reduction can even be carried out by a
  $\AC^0$ circuit family.
\end{proof}

To complement the previous result, let us now consider a very restricted
variety of order $2$ (one can show that all varieties $\vV \subseteq \vLI$ of
order $1$ are so-called \emph{rectangular bands} and contained in $\vLI_1$
already).

\begin{theorem}
  $\SgpInt(\vAp_2 \cap \vNil)$ is $\NP$-complete.%
  \label{thm:nil-bounded}
\end{theorem}
\begin{proof}
  As in the previous proof, we reduce $\SAT$ to $\SgpInt(\vAp_2 \cap \vNil)$.
  Containment in $\NP$ follows from Corollary~\ref{crl:li} and from $\vAp_2
  \cap \vNil \subseteq \vNil \subseteq \vLI$.

  Suppose we are given a set of variables $X = \os{x_1, \dots, x_k}$ as well as
  a set of clauses $\os{C_1, \dots, C_n}$ where $C_j = \os{\ell_{j1},
  \ell_{j2}, \ell_{j3}}$ for each $j \in \os{1, \dots, n}$ and literals
  $\ell_{j1}, \ell_{j2}, \ell_{j3}$ over $X$.
  Let $S$ be the finite semigroup $\set{(i, j)}{1 \le i \le j \le k} \union
  \os{0}$ defined by the multiplication
  \begin{equation*}
    (i, j)(k, \ell) = \begin{cases}
      (i, \ell) & \text{if $k = j+1$}, \\
      0 & \text{otherwise}.
    \end{cases}
  \end{equation*}
  The element $0$ is a zero element.
  Let $g, h_1, \dots, h_n \colon (X \union \overline X)^+ \to S$ be the
  morphisms defined by $g(x_i) = g(\overline{x_i}) = (i, i)$ and by
  \begin{equation*}
    h_j(x_i) =
      \begin{cases}
        (i, i) & \text{if $x_i \not\in C_j$}, \\
        0 & \text{otherwise},
      \end{cases}
    \qquad
    h_j(\overline{x_i}) =
      \begin{cases}
        (i, i) & \text{if $\overline{x_i} \not\in C_j$}, \\
        0 & \text{otherwise}.
      \end{cases}
  \end{equation*}
  for $1 \le i \le k$ and $1 \le j \le n$.
  As accepting sets, we choose $P = \os{(1, k)}$ for $g$ and $Q_1 = \dots = Q_n
  = \os{0}$ for $h_1, \dots, h_n$.
  Again, we would like to show that the intersection
  \begin{equation*}
    L = g^{-1}(P) \intersect \bigcap_{j = 1}^k h_j^{-1}(Q_j)
  \end{equation*}
  is non-empty if and only if there exists a satisfying assignment for $\os{C_1,
  \dots, C_n}$.
  The following two properties hold:
  \begin{enumerate}
    \item $g^{-1}(P)$ contains all words of the form $\ell_1 \cdots \ell_k$
      with $\ell_i \in \os{x_i, \overline{x_i}}$ for $1 \le i \le k$,
    \item $g^{-1}(P) \intersect h_j^{-1}(Q_j)$ contains all words of this form
      containing at least one of the letters $\ell_{j1}, \ell_{j2}, \ell_{j3}$.
  \end{enumerate}

  Let $w \in A^+$ be a word with $g(w) \in P$ and $h_j(w) \in Q_j$ for all $j
  \in \os{1, \dots, n}$. Then, by the first property above, the assignment
  $\mathcal{A}_w$ induced by $w$ is well-defined. Moreover, by the second
  property, we have $\mathcal{A}_w(C_1) = \dots = \mathcal{A}_w(C_n) = 1$ and
  thus, $\mathcal{A}_w$ satisfies $\os{C_1, \dots, C_n}$.
  Conversely, it is easy to see that each word induced by a satisfying
  assignment is contained in $L$.

  Note that the constructed semigroup belongs to $\vAp_2 \cap \vNil$ since by
  definition, we have $(i, j)(i, j) = 0$ for all $(i, j) \in S$.

  It is obvious that the reduction can be performed in polynomial time.
\end{proof}

In view of the previous theorems, the following result might be surprising.
For the class of all commutative semigroups within $\vLI$, the semigroup
intersection problem is $\NP$-hard by Theorem~\ref{thm:unbounded}.
The variety $\vAp_2 \cap \vNil$ has order $2$ and its semigroup intersection
problem is $\NP$-hard by Theorem~\ref{thm:nil-bounded}.
However, if we combine commutativity and bounded order, the problem becomes
easier.

\begin{theorem}
  If $\vV \subseteq \vCom \intersect \vLI$ is a variety of finite semigroups
  with bounded order, then $\vV$ has the {\plPCP}. Thus, there exists some $k
  \in \N$ such that $\SgpInt(\vV) \in \qAC^k$ and $\SgpInt(\vV)$ is decidable
  in quasi-polynomial time.%
  \label{thm:bounded-com}
\end{theorem}
\begin{proof}
  We show that if every monogenic subsemigroup of $S \in \vCom \intersect \vLI$
  has size at most $k$, then every product of at least $k (\log\abs{S} + 1)$
  elements is the zero element. Thus, every non-empty intersection of languages
  recognized by multiple morphisms to such semigroups contains a witness of
  logarithmic size.
  Note that $\vCom \intersect \vLI \subseteq \vNil$, so the $k$-fold power of
  any element in $S$ is the zero element.

  Assume, for the sake of contradiction, that there exists a product of at
  least $k (\log\abs{S} + 1)$ elements which is not the zero element.
  By reordering elements, we can rewrite this product as $s_1^{i_1} \cdots
  s_m^{i_m}$ with $s_i \ne s_j$ for $1 \le i < j \le m$.
  We proceed by induction on $m$.
  If $m \le \log\abs{S} + 1$, then there exists some $r \in \os{1, \dots, m}$
  with $i_r \ge k$.
  Since each monogenic subsemigroup of $S$ has size at most $k$, the element
  $s_r^{i_r}$ then is a zero element, a contradiction. Suppose now that $m >
  \log\abs{S} + 1$.

  The set $T = \mathcal{P}(\os{1, \dots, m}) \setminus \os{\emptyset}$ forms a
  semigroup with union as binary operation.
  Let $h \colon T \to S$ be the morphism defined by $h(r) = s_r^{i_r}$ for $1
  \le r \le m$.
  We have $\abs{T} = 2^m - 1 \ge 2^{m-1} > 2^{\log\abs{S}} = \abs{S}$.
  Thus, by the pigeon hole principle, there exist two sets $K_1, K_2 \subseteq
  \os{1, \dots, m}$ with $K_1 \ne K_2$ and $h(K_1) = h(K_2)$.

  If $K_1 \subsetneq K_2$, then multiplying the product by $h(K_2 \setminus
  K_1)$ does not change its value and $k$-fold multiplication shows that the
  product is zero, a contradiction.
  The case $K_2 \subsetneq K_1$ is symmetric.
  Thus, we may assume that neither $K_1 \subseteq K_2$ nor $K_2 \subseteq K_1$.
  The \emph{length} of a set $K \subseteq \os{1, \dots, m}$ is the sum of all
  $i_r$ with $r \in K$.
  By symmetry, we may assume that the length of $K_1$ is at most the length of
  $K_2$.
  We replace the factor $h(K_1)$ of the product by $h(K_2)$ and obtain the
  statement by induction on the number $m$ of different elements in the
  product\,---\,the length of this new product $h(K_2) h(\os{1, \dots, m}
  \setminus K_1)$ is at least the length of the original product and the number
  of different elements decreases since $K_1 \setminus K_2 \ne \emptyset$.
\end{proof}

\section{Open Problems}

It remains open whether the observation that hardness is not always caused by
purely structural properties also applies to varieties between $\vLI$ and
$\vLDS$ in the semigroup setting, between $\vDO$ and $\vDS$ in the monoid
setting or between $\vR$ and $\vDS$ in the automaton setting.
Another major challenge is obtaining algebraic characterizations of all classes
of finite semigroups with the \pPCP.
As a first step, we suggest proving (or disproving) that the variety $\vLG$ has
the \pPCP.

\paragraph{Acknowledgements.} I would like to thank the anonymous referees of
the conference version of this paper for providing helpful comments.

\newcommand{\Ju}{Ju}\newcommand{\Ph}{Ph}\newcommand{\Th}{Th}\newcommand{\Ch}{Ch}\newcommand{\Yu}{Yu}\newcommand{\Zh}{Zh}\newcommand{\St}{St}\newcommand{\curlybraces}[1]{\{#1\}}

\end{document}